\documentclass[runningheads]{llncs}
\usepackage[T1]{fontenc}
\usepackage{amsmath}
\usepackage{booktabs}
\usepackage{multirow}
\usepackage{subcaption}
\usepackage{hyperref}
\usepackage{amssymb}
\usepackage{wrapfig}
\usepackage{orcidlink}

\usepackage{graphicx}
\begin{document}
\title{L-FUSION: 
Laplacian Fetal Ultrasound Segmentation \& Uncertainty Estimation}
\titlerunning{L-FUSION}
%
\author{Johanna P. M\"uller\inst{1}\orcidlink{0000-0001-8636-7986} \and
Robert Wright\inst{2}\orcidlink{0000-0003-0437-7430} \and
Thomas G. Day\inst{2}\orcidlink{0000-0001-8391-7903} \and
Lorenzo Venturini\inst{2}\orcidlink{0000-0002-7951-6514} \and
Samuel F. Budd\inst{2}\orcidlink{0000-0002-9062-0013} \and
Hadrien Reynaud\inst{3}\orcidlink{0000-0003-0261-2660} \and
Joseph V. Hajnal\inst{2}\orcidlink{000-0002-2690-5495} \and
Reza Razavi\inst{2}\orcidlink{0000-0003-1065-3008} \and
Bernhard Kainz\inst{1,3}\orcidlink{0000-0002-7813-5023}}

\authorrunning{M\"uller et al.}
%
\institute{Friedrich–Alexander University Erlangen–N\"urnberg, DE \\
\email{johanna.paula.mueller@fau.de} \\
\and King's College London, London, UK \and
Imperial College London, London, UK 
}
\maketitle              
\begin{abstract}
Accurate analysis of prenatal ultrasound (US) is essential for early detection of developmental anomalies. However, operator dependency and technical limitations (\emph{e.g.} intrinsic artefacts and effects, setting errors) can complicate image interpretation and the assessment of diagnostic uncertainty.
We present L-FUSION (Laplacian Fetal US Segmentation with Integrated FoundatiON models), a framework that integrates uncertainty quantification through unsupervised, normative learning and large-scale foundation models for robust segmentation of fetal structures in normal and pathological scans. We propose to utilise the aleatoric logit distributions of Stochastic Segmentation Networks and Laplace approximations with fast Hessian estimations to estimate epistemic uncertainty only from the segmentation head. This enables us to achieve reliable abnormality quantification for instant diagnostic feedback.
Combined with an integrated Dropout component, L-FUSION enables reliable differentiation of lesions from normal fetal anatomy with enhanced uncertainty maps and segmentation counterfactuals in US imaging. It improves epistemic and aleatoric uncertainty interpretation and removes the need for manual disease-labelling. Evaluations across multiple datasets show that L-FUSION achieves superior segmentation accuracy and consistent uncertainty quantification, supporting on-site decision-making and offering a scalable solution for advancing fetal ultrasound analysis in clinical settings. Code is available at \url{https://github.com/ividja/L-FUSION}.

\end{abstract}
\section{Introduction}

Fetal ultrasound (US) is a key tool for monitoring development and detecting abnormalities early, with detection rates improving from $15\text{--}59\%$ in the 2000s to $30\text{--}88\%$ today, depending on the lesion type~\cite{10.1002/pd.6613}. However, US imaging remains challenging due to noise, shadows, and anatomical ambiguity. Machine learning models can support clinicians by providing consistent and objective assessments, with detection performance often matching expert evaluation~\cite{Day2024.05.23.24307329,venturini2025whole}.
Yet, standard deterministic models offer no confidence estimates, limiting trust and interpretability in safety-critical settings. Uncertainty quantification helps flag unreliable predictions and guide clinical attention. Current approaches, however, face two main challenges: (1) high computational cost, especially for high-dimensional imaging~\cite{zou2023review}, and (2) the need for large, diverse datasets to capture epistemic uncertainty, critical for identifying unfamiliar or out-of-distribution (OOD) cases~\cite{hullermeier2021aleatoric,NIPS2017_2650d608}.
Laplacian approximations offer a computationally efficient method for capturing both epistemic and aleatoric uncertainty, especially when paired with fast Hessian estimates. Meanwhile, pretrained US foundation models~\cite{jiao2024usfm} address data scarcity by enabling transfer learning across domains. When integrated, these components provide accurate, generalisable segmentation with reliable uncertainty estimates, supporting robust and scalable diagnostic systems for real-world fetal US.
\noindent\textbf{Contributions.}  
We propose \textbf{L-FUSION}, a novel framework for fetal ultrasound segmentation that integrates foundation model embeddings with two complementary uncertainty-aware heads. Our key contributions are:
\noindent\textbf{(1)} We combine a deterministic ultrasound foundation encoder with a Laplacian head for calibrated uncertainty and a Dropout head for counterfactual segmentations, enabling joint aleatoric and epistemic uncertainty quantification. This improves segmentation accuracy and unsupervised OOD detection.
\noindent\textbf{(2)} We reduce computational cost by applying Laplace approximation only to the segmentation head and using efficient Hessian approximations. The frozen encoder and lightweight heads ensure scalability across diverse clinical ultrasound tasks.
\noindent\textbf{(3)} We introduce the first 3D Laplacian segmentation model for ultrasound clips, enabling spatiotemporal extension of our method. In the absence of 3D foundation models, we leverage the 2D encoder as a proxy.

\noindent\textbf{Related Work.}  
Automated ultrasound segmentation lags behind other imaging modalities due to inconsistent metrics, smaller datasets, and high inter- and intra-expert variability~\cite{1661695}. However, growing emphasis on rigorous validation, cross-method comparisons, and larger standardised datasets aims to support clinical uptake~\cite{1661695}. Recent advances incorporate systematic priors, and geometric, temporal, and physics-based constraints to boost accuracy~\cite{ning2021smu,xu2021exploiting,lee2022speckle,gong2023thyroid}.
Foundation models pretrained on large-scale natural images have been adapted for medical imaging, including ultrasound, though cross-modality generalisation remains limited~\cite{ZHANG2024102996}. The Multi-Organ FOundation (MOFO) model~\cite{chen2024multi} enhances segmentation by learning organ-invariant representations and anatomical priors, outperforming single-organ approaches. Similarly, the Universal US Foundation Model (USFM)~\cite{jiao2024usfm} improves generalisation across organs, devices, and centres via self-supervised spatial-frequency dual-masked image modelling on a large, diverse dataset.
While foundation models offer rich features, they lack inherent uncertainty quantification. Bayesian methods provide probabilistic segmentation, estimating epistemic and aleatoric uncertainty~\cite{baumgartner2019phiseg,chen2022medical,gao2023bayeseg}. Laplace approximations efficiently estimate epistemic uncertainty by approximating the posterior, avoiding costly Hessian computation~\cite{NEURIPS2021_a7c95857}. Recent work combines aleatoric logit distributions from Stochastic Segmentation Networks~\cite{monteiro2020stochastic} with Laplace approximations for comprehensive uncertainty quantification~\cite{10.1007/978-3-031-72111-3_33}.

\begin{figure}[h]
    \centering
    \includegraphics[width=\textwidth, trim=0 220 0 0, clip]{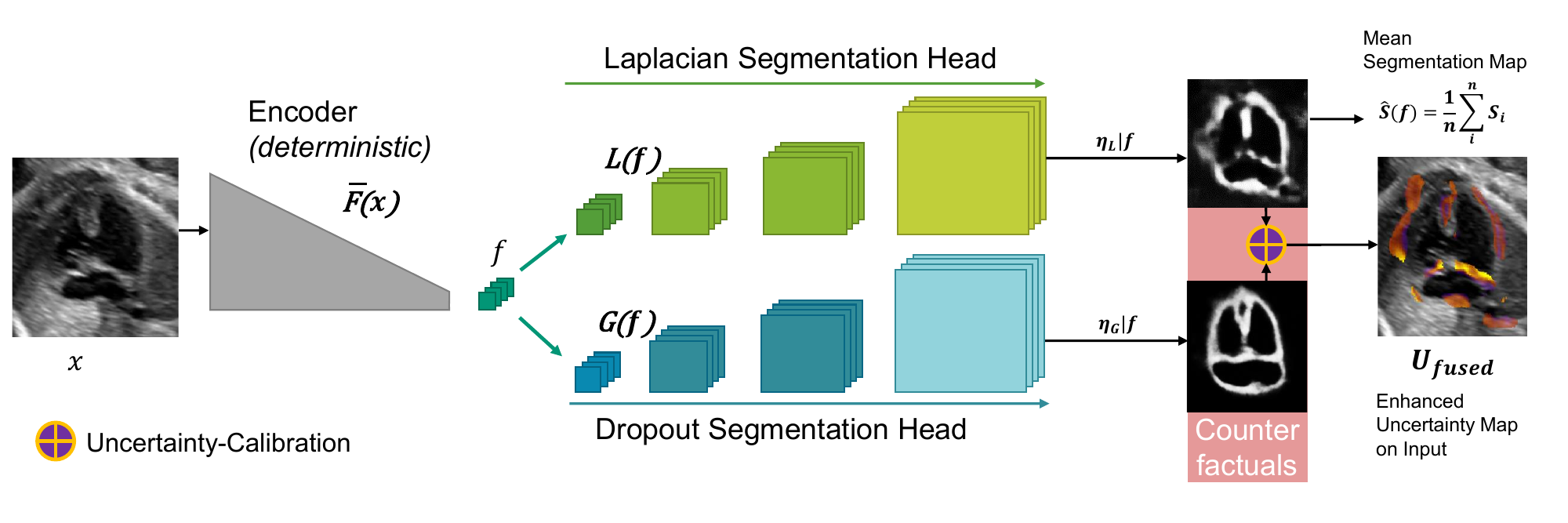}
    \caption{Foundation model (Encoder) with Laplacian segmentation head $L(f)$ and a Dropout segmentation path. Here, $L(f)$ and Dropout Segmentation Head $G(f)$ decoding the embeddings $f = \bar{F}(x)$. The segmentations $S_i(f)$ are predicted by $L(f)$. We enhance the uncertainty measure maps $U_{L/D,k}$ with Uncertainty-Calibration as $U_{fused}$ by including the variance and total entropy of logits.}
    \label{fig:laplace-scheme}
\end{figure}

\section{Method}
We introduce L-FUSION, a probabilistic segmentation framework for fetal ultrasound that combines a frozen foundation model, \emph{e.g.}, USFM~\cite{jiao2024usfm}, with two complementary uncertainty-aware segmentation heads: a Laplacian-based branch and a Dropout-based branch~\cite{gal2016dropout}. By jointly modelling both epistemic and aleatoric uncertainty, L-FUSION produces more reliable predictions and highlights regions of ambiguity, particularly useful for clinical interpretation and decision-making. An overview of the framework is shown in Fig.~\ref{fig:laplace-scheme}.

\noindent\textbf{Feature Encoding} L-FUSION begins by passing an input $2D$ ultrasound image $x$ through a pretrained foundation encoder $\bar{F}$, \emph{e.g.}, USFM~\cite{jiao2024usfm}. The encoder remains frozen during training to preserve the generalisable representations learned during pretraining and to reduce the computational burden associated with fine-tuning the encoder as well. The output feature embedding $f = \bar{F}(x)$ is shared by both segmentation branches. This design choice decouples representation learning from downstream uncertainty estimation, ensuring that all uncertainty arises from the segmentation heads rather than encoder variability.

\begin{theorem}[Sufficiency of Pretrained Embeddings for Segmentation and Uncertainty Quantification] 
Let $x \in \mathcal{X}$ be an input image, and let $\bar{F}: \mathcal{X} \to \mathcal{F}$ be a deterministic, pretrained foundation model mapping $x$ to a feature embedding $f = \bar{F}(x)$. Consider a Laplacian segmentation head $L$ that models the posterior segmentation distribution as  
$p(S | f, \theta) = \prod_{i} \text{Bernoulli}(s_i | \sigma(\eta_{L,i})),$ 
where the latent logits $\eta$ follow a Laplace-approximated posterior:  
$p(\eta | f, \theta) = \mathcal{N}(\eta | \mu_\theta(f), PP^\top + D),$
with $\mu_\theta(f)$ the mean prediction, and $PP^\top + D$ the covariance matrix.  
If $\bar{F}$ preserves all segmentation-relevant information, i.e., using the mutual information
$I(S; x) = I(S; f),$  
then $f$ is a \emph{sufficient statistic} for estimating segmentation and associated predictive uncertainty.  
\end{theorem}

\begin{proof} 
Since $\bar{F}$ is deterministic, the distribution over $f$ given $x$ reduces to a delta function:
$p(f | x) = \delta(f - \bar{F}(x)),$
so the segmentation posterior satisfies
$p(S | x) = p(S | f, \theta).$
By the Data Processing Inequality~\cite{thomas2006elements},  
$I(S; x) \geq I(S; f).$
We assume $I(S; x) = I(S; f),$ i.e., that no information about $S$ is lost in passing from $x$ to $f$. This implies that $f$ is a \emph{sufficient statistic} for $S$ with respect to $x$.
Now, consider uncertainty estimation. The Laplace approximation assumes a Gaussian posterior over parameters:  
$q(\theta^*) = \mathcal{N}(\theta^* | \theta_{\text{MAP}}, \mathbf{H}^{*-1}),$
where $\mathbf{H}^*$ is the Hessian of the negative log-likelihood at the MAP estimate. Predictive uncertainty is then propagated via linearisation of the model output $\eta$ around $\theta_{MAP}$, yielding:
$\operatorname{Var}(S | f) \approx J \mathbf{H}^{*-1} J^\top + \sigma^2(f) I,$
where $J$ is the Jacobian of $\mu_\theta(f)$ for $\theta$, evaluated at $\theta_{MAP}$, and $\sigma^2(f)$ is the predictive variance of the likelihood model (e.g., from $D$ in the Laplace covariance).
The first term corresponds to epistemic uncertainty, and the second to aleatoric uncertainty. Since both terms depend only on $f$, and $f$ is sufficient for $S$, it follows that $f$ is also sufficient for segmentation uncertainty estimation.
\end{proof}

\noindent\textbf{Laplacian Segmentation Path.}
The first segmentation branch estimates uncertainty using a Laplace approximation centred on the MAP estimate of parameters $\theta_{\text{MAP}}$. Given embedding $f$ from the frozen encoder, the Laplacian segmentation head $L$ models logits as a Gaussian distribution, $\eta_L \mid f \sim \mathcal{N}(\mu(f, \theta_1), $ $\Sigma(f, \theta_2))$, where $\theta_1$ and $\theta_2$ govern the mean and covariance.
Segmentation masks $S_i$ are sampled from Bernoulli distributions parameterised by the sigmoid of each logit $\eta_{L,i}$, for $i = 1, \ldots, n$ samples per input. To capture aleatoric uncertainty from data noise and variability, Monte Carlo (MC) sampling~\cite{shapiro2003monte} is performed over $m$ segmentation networks $L_j(f)$, $j = 1, \ldots, m$, drawn from the Laplace approximation $q(\theta^*)$ around the MAP estimate, using embedding $f = \bar{F}(x)$.
Each sampled network outputs a logits distribution $p(\eta \mid f, \theta)$. Aggregating predictions yields the mean segmentation map $\hat{S}_L$ and uncertainty measures $U_{L,k}$, $k=1, \ldots, K$. 
$U_{L,k}^E$ quantifies epistemic uncertainty via mutual information $\mathcal{I}(p(S, \theta^* \mid f, \mathcal{F}))$ over embeddings $\mathcal{F}$, highlighting prediction unreliability. Conversely, $U_{L,k}^A$ captures aleatoric uncertainty by computing expected entropy $\mathcal{H}(p(S \mid f, \theta^*))$, reflecting noise inherent in imaging or anatomy.
For efficiency, the Laplace approximation uses a low-rank Hessian parametrisation. With low-rank matrix $P$, the predictive distribution is $p(\eta \mid f, \theta) = \mathcal{N}(\eta \mid \mu_\theta, PP^\top + D)$, where $D$ is diagonal. This follows methods in~\cite{monteiro2020stochastic,10.1007/978-3-031-72111-3_33}.

\noindent\textbf{Dropout Segmentation Path.}
To complement the Laplacian segmentation head, we introduce a Dropout-based segmentation path. The Laplace approximation fits a Gaussian distribution over $\theta_{\text{MAP}}$, producing a Gaussian predictive distribution for segmentation logits and masks. This yields well-calibrated but unimodal uncertainty estimates, reflecting smooth local variations from parameter uncertainty near the mode~\cite{graves2011practical}. However, it may miss multiple plausible segmentation hypotheses or modes, especially with ambiguous or OOD inputs.
In contrast, Dropout adds stochasticity by randomly dropping units, effectively sampling from a richer, more diverse space of model configurations, enabling a multimodal predictive distribution over segmentations. Thus, we employ a Dropout Segmentation Head $G(f)$ (Fig.~\ref{fig:laplace-scheme}), based on MC-dropout~\cite{kendall2015bayesian}, which takes embeddings from the Foundation Model encoder as input. From this path, we obtain the logit distribution $\eta_{\text{MC}}$ for sampling logits and a set of uncertainty measures $U_{D,k}$.

\noindent\textbf{Counterfactuals and Uncertainty Calibration.}
We propose a hybrid uncertainty estimation strategy that fuses outputs from two complementary segmentation heads: Laplacian and Dropout. The Laplacian head produces calibrated, unimodal predictions near the MAP estimate, while the Dropout head captures diverse plausible alternatives as natural counterfactuals.
Uncertainty is quantified via intra- and inter-path logit variance, measuring disagreement within and between heads. To emphasise clinically meaningful ambiguity, this disagreement is scaled by total entropy through a Hadamard product, suppressing low-signal noise and highlighting regions where both variance and entropy indicate instability, especially under OOD conditions.
Though Dropout can be overconfident in OOD scenarios~\cite{ovadia2019can}, we leverage this by fusing it with the Laplacian path to enhance OOD sensitivity. The fused uncertainty map is defined as
$ U_{\text{fused}} = \sqrt{\mathrm{Var}(\eta^{L} - \eta^{D})} \circ \mathcal{H}(\eta^{L} - \eta^{D}) \cdot \mathrm{Var}(\tilde{U}_{L/D,k}),$
where $\eta^L$ and $\eta^D$ are logits from the Laplacian and Dropout paths, respectively, and $\tilde{U}_{L/D,k}$ is an auxiliary pixel-wise uncertainty measure (e.g., mutual information). Here, $\mathrm{Var}(\tilde{U}_{L/D,k})$ captures global pixel-level uncertainty dispersion as a weighting factor.
This formulation combines local disagreement, entropy, and global variation, amplifying signals in clinically ambiguous regions while suppressing noise. The fusion yields interpretable uncertainty maps that aid clinical decision-making by highlighting low-confidence areas from epistemic or aleatoric uncertainty, helping sonographers decide when to re-scan, refer, or flag abnormalities in challenging fetal ultrasound cases.

\section{Evaluation}
\noindent\textbf{Datasets.}  
We use two fetal ultrasound datasets: 
\textit{Fetal Heart (4CH)}: Our dataset includes $428$ normal fetal four-chamber (4CH) scans for training, annotated with six classes (heart chambers, heart, thorax). The test set contains $92$ normal and $193$ abnormal samples with Atrioventricular Septal Defect (AVSD), used for OOD evaluation.
\textit{HC18}~\cite{van2018automated}: A public dataset of 1,334 2D fetal head ultrasound images used for head circumference estimation. It includes $999$ annotated training images and $335$ unannotated test images, available via
(\href{DOI:10.5281/zenodo.1322001}{Zenodo}). OOD samples for HC18 are synthetically generated using elastic deformations to simulate local skull shape variations. Additionally, a synthetic variant of the HC18 dataset, denoted HC18$^\star$, is constructed by simulating frame-by-frame ultrasound sweeps using continuous affine transformations to emulate realistic temporal acquisition dynamics. Each sequence consists of 16 frames.

\noindent\textbf{Implementation and Training.}
We use USFM~\cite{jiao2024usfm}, pretrained via self-super-vised Masked Image Modelling (MIM) on the 3M-US dataset without labels or annotations. We extended the $nnj$ library~\cite{software:nnj} with BatchNorm to prevent exploding gradients and support higher-dimensional layers (\textit{Conv3d, Upsample3d, BatchNorm3d}). Images are resized to $256 \times 256$ for U-Nets and $224 \times 224$ for the USFM encoder~\cite{jiao2024usfm}. All U-Nets share the same architecture except Dropout U-Net~\cite{kendall2015bayesian}, which adds Dropout after each convolutional block. Training uses linear augmentations and random crops up to $80\%$. Fetal heart scans are aligned via an iterative 2D method~\cite{wright2023fast}. Models trained on an NVIDIA A100 80 GB with early stopping.

\noindent\textbf{Metrics.}  
Segmentation is evaluated using Dice (DSC), Hausdorff Distance (HD), and Absolute Difference (AD; average absolute error) for HC18, and Intersection over Union (IoU) for fetal heart datasets, measuring overlap and boundary accuracy. For OOD detection, following \cite{10.1007/978-3-031-72111-3_33}, we report AUROC/AUC over epistemic measures—Mutual Information (MI), Expected Pairwise KL (EP), Pixel Variance (PV)—and aleatoric uncertainty via Expected Entropy (EE).

\begin{table}[h!]
    \centering
    \caption{Results on Head Circumference Dataset HC18; DSC [\%] - Dice Similarity Coefficient, HD [mm] - Hausdorff Distance, AD [mm] - Absolute Difference; \textbf{1st-ranked}, \underline{2nd-ranked}.}
    \label{tab:hc18-seg}
    \resizebox{\textwidth}{!}{ %
    \begin{tabular}{lccccccccc|ccc}
        & \multicolumn{3}{c}{Mean} & \multicolumn{3}{c}{25-Percentile} & \multicolumn{3}{c}{75-Percentile} & \multicolumn{3}{c}{1. Trimester} \\
        \cmidrule(lr){2-4} \cmidrule(lr){5-7} \cmidrule(lr){8-10}\cmidrule(lr){11-13}
        Model    & DSC  $\uparrow$ & HD  $\downarrow$ & AD  $\downarrow$ & DSC  $\uparrow$ & HD  $\downarrow$ & AD  $\downarrow$ & DSC  $\uparrow$ & HD  $\downarrow$ & AD  $\downarrow$ & DSC  $\uparrow$ & HD  $\downarrow$ & AD  $\downarrow$ \\
        \cmidrule(lr){2-2} \cmidrule(lr){3-3} \cmidrule(lr){4-4} \cmidrule(lr){5-5} \cmidrule(lr){6-6} \cmidrule(lr){7-7}\cmidrule(lr){8-8} \cmidrule(lr){9-9} \cmidrule(lr){10-10}\cmidrule(lr){11-11} \cmidrule(lr){12-12} \cmidrule(lr){13-13}
        SSN~\cite{monteiro2020stochastic}  & $87.6_{\pm 15.0}$ & $8.4_{\pm 9.1}$  & $16.7_{\pm 22.5}$  & 85.7 & 1.9 & \underline{2.9} & 96.6 & 10.9 & 19.9  & $73.2_{\pm 23.5}$   & $13.7_{\pm 9.6}$  & $23.3_{\pm 29.6}$  \\
        Drop. U-Net~\cite{kendall2015bayesian}  & $\underline{94.3_{\pm 8.0}}$   & $\underline{4.1_{\pm 4.3}}$   & $\underline{9.2_{\pm 11.3}}$  & \underline{93.9} & \underline{1.8} & 3.2 & \underline{97.2} & \underline{4.9} & \underline{10.5} & $\underline{87.7_{\pm 17.6}}$ & $\underline{4.2_{\pm 7.3}}$  & $\underline{10.1_{\pm 21.0}}$ \\   
        Lap. U-Net~\cite{10.1007/978-3-031-72111-3_33}   & $85.4_{\pm 16.9}$   & $9.9_{\pm 11.7}$   & $21.9_{\pm 31.6}$   & 82.7 & 2.7 & 4.3 & 96.3 & 11.6 & 25.2 & $70.3_{\pm 22.2}$  & $11.3_{\pm 10.6}$   & $29.7_{\pm 30.1}$  \\
        L-FUSION (ours)       & $\mathbf{95.5_{\pm 8.0}}$  & $\mathbf{2.7_{\pm 3.0}}$  & $\mathbf{4.7_{\pm 6.2}}$  & $\mathbf{95.9}$ & $\mathbf{1.3}$ & $\mathbf{1.4}$ &$ \mathbf{97.7}$ & $\mathbf{2.9}$ & $\mathbf{5.6}$ & $\mathbf{89.0_{\pm 17.7}}$  & $\mathbf{3.4_{\pm 5.6}}$  & $\mathbf{6.4_{\pm 10.8}}$  \\
    \end{tabular}
    }
\end{table}

\noindent\textbf{Quantitive Results.}
L-FUSION outperforms all baselines on HC18 (Tab.~\ref{tab:hc18-seg}), achieving the highest DSC ($95.5\%$) and lowest HD ($2.7~mm$) and AD ($4.7~mm$), surpassing Dropout U-Net by $1.3\%$ DSC and reducing HD by $34.1\%$. Gains persist across percentiles and the first trimester, with DSC up $1.5\%$ and AD down $36.6\%$. It also shows the lowest standard deviations and best percentile metrics overall.
For fetal heart ultrasound (Tab.~\ref{tab:segmentation-metrics}), L-FUSION sets new benchmarks. In unaligned 4CH, it exceeds the best baseline by $12.8\%$ DSC and cuts HD by $41.6\%$. With thorax- and heart-alignment, it leads by $9.5\%$ and $0.9\%$ DSC, consistently achieving the lowest HD. Training time drops by up to $40\%$, with greater savings under stronger alignment, confirming its state-of-the-art status in fetal heart and head circumference segmentation.

\begin{table}[h!]
    \centering
    \caption{Segmentation Performance on Fetal Hearts over all class labels; Plane - Fetal heart US standard plane; \textit{A.} - Aligned to Thorax (T) or Heart (H); DSC [\%] - Dice Similarity Coefficient, IoU [\%] - Intersection over Union, HD [pixel] - Hausdorff Distance; \textbf{1st-ranked}, \underline{2nd-ranked}.}
    \label{tab:segmentation-metrics}
    \resizebox{1.0\columnwidth}{!}{
    \begin{tabular}{llccccccccccccc}
         & & & \multicolumn{3}{c}{SSN~\cite{monteiro2020stochastic} } & \multicolumn{3}{c}{Dropout U-Net~\cite{kendall2015bayesian}} & \multicolumn{3}{c}{Laplace U-Net}~\cite{10.1007/978-3-031-72111-3_33} & \multicolumn{3}{c}{L-FUSION (ours)} \\
        \cmidrule(lr){4-6} \cmidrule(lr){7-9} \cmidrule(lr){10-12} \cmidrule(lr){13-15}
                                  & Plane & \textit{A.} & DSC $\uparrow$   & IoU $\uparrow$   & HD $\downarrow$   & DSC $\uparrow$   & IoU $\uparrow$   & HD  $\downarrow$  & DSC $\uparrow$   & IoU $\uparrow$   & HD $\downarrow$   & DSC $\uparrow$   & IoU $\uparrow$   & HD $\downarrow$   \\
        \cmidrule(lr){4-4} \cmidrule(lr){5-5} \cmidrule(lr){6-6} \cmidrule(lr){7-7} \cmidrule(lr){8-8} \cmidrule(lr){9-9} \cmidrule(lr){10-10} \cmidrule(lr){11-11} \cmidrule(lr){12-12} \cmidrule(lr){13-13} \cmidrule(lr){14-14}\cmidrule(lr){15-15}
        \multirow{3}{*}{2D} & 4CH &  -  & $40.6_{\pm 10.4}$ & $26.1_{\pm8.8}$ &$ \underline{44.0_{\pm 7.1}}$  & $43.0_{\pm 9.0}$ & $27.9_{\pm 7.7}$ & $62.3_{\pm 16.9}$ & $\underline{43.1_{\pm 10.1}}$ & $\underline{28.0_{\pm 8.9}}$ & $49.8_{\pm9.0}$ & $\mathbf{55.9_{\pm 8.6}}$ & $\mathbf{39.3_{\pm 9.1}}$ & $\mathbf{29.1_{\pm 4.7 }}$ \\
        & 4CH          &   T  &  $47.7_{\pm 11.7}$ & $32.1_{\pm 10.4}$ &$ 63.4_{\pm 12.7}$  & $60.2_{\pm 5.9}$ & $43.3_{\pm 6.3}$ & $89.5_{\pm 42.3}$ & $\underline{61.2_{\pm 7.6}}$ & $\underline{44.5_{\pm 7.8}}$ & $\underline{37.0_{\pm 10.7}}$ & $\mathbf{67.0_{\pm 7.9}}$ & $\mathbf{50.9_{\pm 8.5}}$ & $\mathbf{36.9_{\pm 15.1 }}$ \\
        & 4CH          &  H   & $52.4_{\pm 12.2}$ & $36.5_{\pm 12.2}$ &$ \underline{56.1_{\pm 16.7}}$  & $\underline{65.0_{\pm 6.7}}$ & $\underline{48.5_{\pm7.8}}$ & $68.1_{\pm 26.0}$ & $55.9_{\pm 9.3}$ & $39.5_{\pm 9.9}$ & $56.8_{\pm 14.3}$ & $\mathbf{65.6_{\pm 6.8}}$ & $\mathbf{49.2_{\pm 7.9}}$ & $\mathbf{35.0_{\pm 10.7}}$ \\
        \midrule
        \multirow{2}{*}{Clip} & 4CH  & & $46.1_{\pm 12.4}$  & $30.9_{\pm 11.9}$ & $\underline{23.7_{\pm 4.9}}$ & $33.7_{\pm 12.9}$ & $21.1_{\pm 10.3}$ & $83.8_{\pm 32.4}$ & $\underline{53.7_{\pm 7.8}}$ & $\underline{37.1_{\pm 7.2}}$  & $23.8_{\pm 5.0}$ & $\mathbf{66.8_{\pm 8.1}}$  & $\mathbf{51.3_{\pm 7.5}}$ & $\mathbf{23.4_{\pm 5.2}}$  \\
        & 4CH        &   H  & $\underline{68.1_{\pm5.9}}$  & $\underline{51.9_{\pm6.6}}$ & $\mathbf{20.8_{\pm 3.7}}$ & $67.4_{\pm6.0 }$ & $51.1_{\pm 6.5}$ & $32.2_{\pm 15.6}$ & $\mathbf{71.2_{\pm 5.3}}$ & $\mathbf{55.5_{\pm 6.3}}$  & $\underline{21.3_{\pm 3.9}}$ & $65.6_{\pm 6.4}$  & $49.1_{\pm 6.9}$ & $22.3_{\pm 3.9}$ \\
    \end{tabular}}
\end{table}

\noindent Tab.~\ref{tab:ood_auc} compares top individual OOD measures for epistemic and aleatoric uncertainty against Uncertainty-Calibration, which scales predicted logit variance ($n=50$) to improve OOD detection AUC. Applied intra- or inter-path, Uncertain-ty-Calibration significantly boosts AUC. Ablation shows inter-path calibration raises AUC by $19\%$ (HC18) and $20\%$ (4CH) on unaligned data; for aligned 4CH, gains are $3$–$9\%$. Segmentation runs at a minimum 45 FPS. The best uncertainty measure is identified once at deployment or major data shifts via held-out testing. Pixel Variance, Mutual Information, and Expected Entropy emerge as most informative. With calibrated maps and $n=20$ samples, inference speed stays above 40 FPS.

\begin{table}[h!]
    \centering
    \caption{AUC for OOD Detection with Uncertainty-Calibration (identification of best-scoring uncertainty measure $\tilde{U}_{L/D,k}$): $^\circ$Expected Entropy, $^\diamond$Expected Pairwise Kullback-Leibler, $^\ast$Pixel Variance, $^\dagger$Mutal Information); \textit{A.} - Aligned to Thorax (T) or Heart (H); 
    D. - Dropout, U. - U-Net, L. - Laplace, (C-1) UMSF $+$ D. Segmentation Head, (C-2) UMSF $+$ D.U.; \textbf{1st-ranked}, \underline{2nd-ranked}.}
    \label{tab:ood_auc}
        \resizebox{0.9\textwidth}{!}{ 
    \begin{tabular}{llcccccccccccc}
        & & & & & & & \multicolumn{6}{c}{Uncertainty - Calibration} \\
        \cmidrule{8-13}
        & && \multicolumn{4}{c}{Indiv. OOD Measures}& \multicolumn{4}{c}{Intra-Path} & \multicolumn{2}{c}{Inter-Path} \\
        \cmidrule(lr){4-7}\cmidrule(lr){8-11}\cmidrule(lr){12-13}
        \multicolumn{2}{l}{Data} & \textit{A.} & D.U.\cite{kendall2015bayesian} & L.U.\cite{10.1007/978-3-031-72111-3_33} & L.\cite{jiao2024usfm} & D.\cite{jiao2024usfm}  &  D.U.\cite{kendall2015bayesian} & L.U.\cite{10.1007/978-3-031-72111-3_33} & L.\cite{jiao2024usfm} & D.\cite{jiao2024usfm} &  C-1 & C-2 \\
         \cmidrule(lr){4-4} \cmidrule(lr){5-5} \cmidrule(lr){6-6}\cmidrule(lr){7-7}\cmidrule(lr){8-8} \cmidrule(lr){9-9}\cmidrule(lr){10-10}\cmidrule(lr){11-11}\cmidrule(lr){12-12}\cmidrule(lr){13-13}
        \multirow{4}{*}{2D} 
        & HC18&& $0.74^{\dagger}$ & $0.57^{\dagger}$ & $0.70^{\dagger}$ & $0.72^{\circ}$ & $0.84^{\ast}$ & $0.52^{\dagger}$& $0.71^{\dagger}$ & $0.74^{\dagger}$ & $\mathbf{1.00}^{\ast}$ & $\mathbf{1.00}^{\ast}$\\
        & 4CH && $0.61^{\circ}$ & $0.56^{\circ}$ & $0.57^{\dagger}$ & $0.61^{\circ}$ & $0.61^{\circ}$ & $0.51^{\circ}$ & $0.54^{\ast}$ & $0.67^{\ast}$ & $\mathbf{0.84}^{\ast}$ & $\underline{0.73}^{\ast}$\\
        & 4CH &T& $0.61^{\circ}$ & $0.54^{\circ}$ & $0.56^{\circ}$ & $0.56^{\circ}$& $\mathbf{0.89}^{\ast}$ & $0.53^{\ast}$ & $0.54^{\ast}$&  $0.72^{\ast}$ & $0.64^{\ast}$ & $\underline{0.82}^{\dagger}$\\
        & 4CH &H& $0.64^{\diamond}$ & $0.59^{\dagger}$ & $0.55^{\ast}$ & $0.56^{\ast}$& $\mathbf{0.85}^{\circ}$ &  $0.54^{\ast}$ & $0.51^{\circ}$& $0.64^{\circ}$ & $0.71^{\ast}$ & $\underline{0.82}^{\dagger}$ \\
        \midrule
        \multirow{3}{*}{Clip} 
        & HC18$^\star$ && $0.73^{\circ}$ & $0.57^{\dagger}$ & $0.69^{\dagger}$ & $0.70^{\circ}$ & $0.82^{\ast}$ & $0.53^{\dagger}$ & $0.73^{\circ}$ & $0.75^{\dagger}$ & $\underline{0.98^{\ast}}$ & $\mathbf{0.99^{\ast}}$\\
        & 4CH && $\mathbf{0.88^{\dagger}}$ & $0.77^{\dagger}$ & $0.71^{\ast}$ & $\underline{0.81^{\diamond}}$ & $\underline{0.81^{\circ}}$ & $0.58^{\dagger}$ & $0.55^{\circ}$ & $\mathbf{0.88^{\circ}}$& $\mathbf{0.88^{\circ}}$ & $\underline{0.81^{\circ}}$ \\
        & 4CH &H& $\underline{0.67^{\dagger}}$ & $\underline{0.67^{\circ}}$ & $0.66^{\circ}$ & $0.62^{\ast}$ & $\underline{0.67^{\circ}}$& $0.59^{\circ}$ & $\underline{0.67^{\circ}}$ & $\mathbf{0.73^{\ast}}$ & $\mathbf{0.73^{\ast}}$ & $\underline{0.67^{\dagger}}$ \\
    \end{tabular}}
\end{table}
\begin{figure}[h!]
    \centering
    \includegraphics[width=0.95\textwidth]{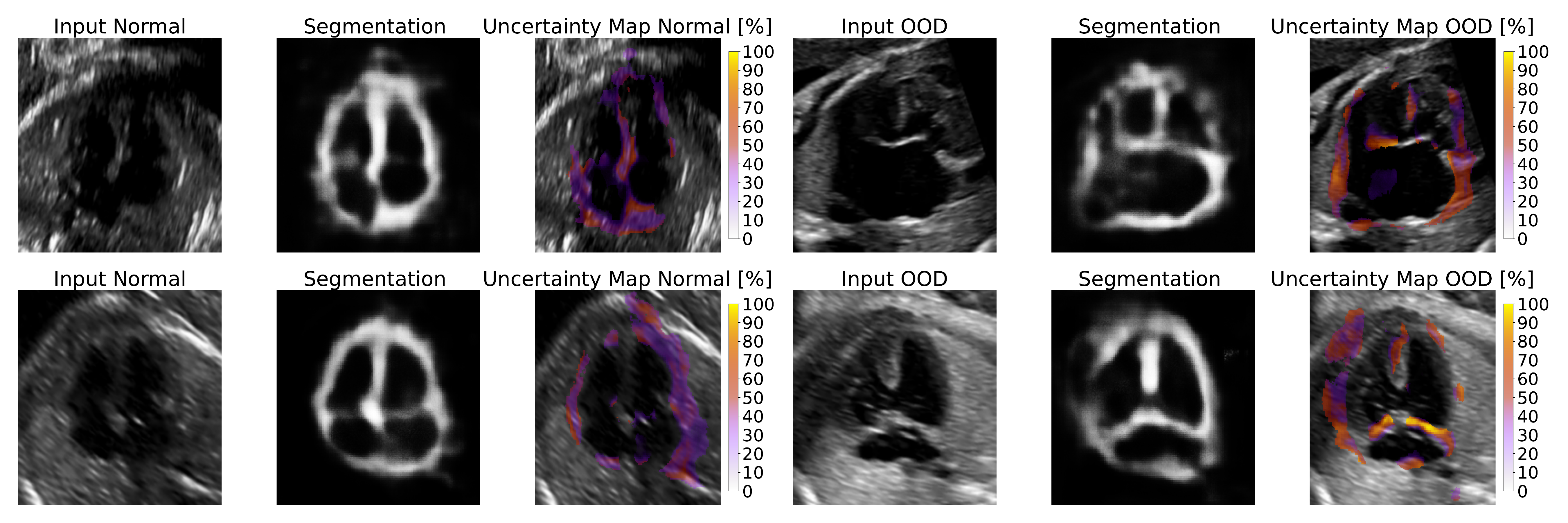}
    \caption{Segmentation and Enhanced Uncertainty maps $U_{fused}$ by L-FUSION for Normal and OOD (AVSD) cases in four-chamber view (4-CH) Heart-aligned; the segmented class corresponds to myocardium and valvular structures.Elevated uncertainty levels indicate increased model uncertainty, which may correspond to abnormalities or deviations from the normal training distribution.}
    \label{fig:quality_comparison}
\end{figure}

\noindent\textbf{Qualitative Results.}  
Dropout-based uncertainty maps are diffuse and over-segmented, missing fine details and limiting epistemic uncertainty capture (Fig.~\ref{fig:laplace-scheme}). Combining the foundation model with a Laplacian head improves segmentation and uncertainty detail but increases global uncertainty. L-FUSION fuses inter-path logit variance to enable counterfactual analysis and produce refined uncertainty maps (Fig.~\ref{fig:quality_comparison}). In both OOD cases, these maps reveal defects in septal formation and valve abnormalities, confirmed by diagnosis.

\noindent\textbf{Discussion.}  
Our method delivers efficient segmentation with fine-grained uncertainty quantification by combining Laplacian approximation and MC Dropout. With Uncertainty-Calibration, OOD detection AUC improves by up to $+20\%$. However, performance may drop when domain-specific features dominate or data lies on complex manifolds. Uncertainty estimates also rely on embedding quality; loss of local detail or high noise can degrade segmentation.

\section{Conclusion}  
Ultrasound segmentation has advanced but still lags behind CT and MRI in uncertainty quantification and consistency~\cite{1661695}. Although foundation models~\cite{chen2024multi,jiao2024usfm} boost accuracy, they often neglect uncertainty, crucial for clinical safety. Our approach, combining a foundation encoder with a Laplacian head using fast Hessian approximation, a Dropout unit, and uncertainty calibration, reduces model complexity while improving performance. It enhances segmentation accuracy and produces localised uncertainty maps, increasing reliability, aiding abnormality detection, and supporting safer clinical decisions. This superior uncertainty quantification surpasses deterministic and Dropout U-Nets, helping sonographers identify low-confidence areas and flag abnormalities for expert review.

\begin{credits}
\subsubsection{\ackname}
Support was received from the ERC project MIA-NORMAL 101083647, the State of Bavaria (HTA) and DFG 512819079. HPC resources were provided by NHR@FAU of FAU Erlangen-N\"urnberg under the NHR project b180dc. NHR@FAU hardware is partially funded by the DFG – 440719683. We thank the volunteers and sonographers at St. Thomas’ Hospital London.
We  gratefully acknowledge financial support from the Wellcome Trust IEH 102431, EPSRC (EP/S022104/1, EP/S013687/1), EPSRC Centre for Medical Engineering [WT 203148/Z/16/Z], the National Institute for Health Research (NIHR) Biomedical Research Centre (BRC) based at Guy’s and St Thomas’ NHS Foundation Trust and King’s College London and supported by the NIHR Clinical Research Facility (CRF) at Guy’s and St Thomas’. This work was also supported by the UK Research and Innovation AI Centre for Value Based Healthcare.
\end{credits}

%
%
%
\bibliographystyle{splncs04}
\bibliography{bibliography}

\begin{thebibliography}{10}
\providecommand{\url}[1]{\texttt{#1}}
\providecommand{\urlprefix}{URL }
\providecommand{\doi}[1]{https://doi.org/#1}

\bibitem{baumgartner2019phiseg}
Baumgartner, C.F., Tezcan, K.C., Chaitanya, K., H{\"o}tker, A.M., Muehlematter, U.J., Schawkat, K., Becker, A.S., Donati, O., Konukoglu, E.: Phiseg: Capturing uncertainty in medical image segmentation. In: Medical Image Computing and Computer Assisted Intervention--MICCAI 2019: 22nd International Conference, Shenzhen, China, October 13--17, 2019, Proceedings, Part II 22. pp. 119--127. Springer (2019)

\bibitem{chen2024multi}
Chen, H., Cai, Y., Wang, C., Chen, L., Zhang, B., Han, H., Guo, Y., Ding, H., Zhang, Q.: Multi-organ foundation model for universal ultrasound image segmentation with task prompt and anatomical prior. IEEE Transactions on Medical Imaging  (2024)

\bibitem{chen2022medical}
Chen, X., Zhao, Y., Liu, C.: Medical image segmentation using scalable functional variational bayesian neural networks with gaussian processes. Neurocomputing  \textbf{500},  58--72 (2022)

\bibitem{NEURIPS2021_a7c95857}
Daxberger, E., Kristiadi, A., Immer, A., Eschenhagen, R., Bauer, M., Hennig, P.: Laplace redux - effortless bayesian deep learning. In: Ranzato, M., Beygelzimer, A., Dauphin, Y., Liang, P., Vaughan, J.W. (eds.) Advances in Neural Information Processing Systems. vol.~34, pp. 20089--20103. Curran Associates, Inc. (2021)

\bibitem{Day2024.05.23.24307329}
Day, T.G., Matthew, J., Budd, S.F., Farruggia, A., Venturini, L., Wright, R., Jamshidi, B., To, M., Ling, H., Lai, J., Tan, M.Y., Brown, M., Guy, G., Casagrandi, D., Arechvo, A., Syngelaki, A., Lloyd, D., Zidere, V., Vigneswaran, T., Miller, O., Akolekar, R., Nanda, S., Nicolaides, K., Kainz, B., Simpson, J.M., Hajnal, J.V., Razavi, R.: Artificial intelligence to assist in the screening fetal anomaly ultrasound scan (prometheus): A randomised controlled trial. medRxiv, to appear in NEJM AI  (2024). \doi{10.1101/2024.05.23.24307329}

\bibitem{10.1002/pd.6613}
Freud, L.R., Simpson, L.L., Wilkins-Haug, L.E.: The bright future of fetal cardiology. Prenatal Diagnosis  \textbf{44}(6-7),  676--678 (2024). \doi{https://doi.org/10.1002/pd.6613}, \url{https://obgyn.onlinelibrary.wiley.com/doi/abs/10.1002/pd.6613}

\bibitem{gal2016dropout}
Gal, Y., Ghahramani, Z.: Dropout as a bayesian approximation: Representing model uncertainty in deep learning. International Conference on Machine Learning  (2016)

\bibitem{gao2023bayeseg}
Gao, S., Zhou, H., Gao, Y., Zhuang, X.: Bayeseg: Bayesian modeling for medical image segmentation with interpretable generalizability. Medical Image Analysis  \textbf{89},  102889 (2023)

\bibitem{gong2023thyroid}
Gong, H., Chen, J., Chen, G., Li, H., Li, G., Chen, F.: Thyroid region prior guided attention for ultrasound segmentation of thyroid nodules. Computers in biology and medicine  \textbf{155},  106389 (2023)

\bibitem{graves2011practical}
Graves, A.: Practical variational inference for neural networks. Advances in neural information processing systems  (2011)

\bibitem{van2018automated}
van~den Heuvel, T.L., de~Bruijn, D., de~Korte, C.L., Ginneken, B.v.: Automated measurement of fetal head circumference using 2d ultrasound images. PloS one  \textbf{13}(8),  e0200412 (2018)

\bibitem{hullermeier2021aleatoric}
H{\"u}llermeier, E., Waegeman, W.: Aleatoric and epistemic uncertainty in machine learning: An introduction to concepts and methods. Machine learning  \textbf{110}(3),  457--506 (2021)

\bibitem{jiao2024usfm}
Jiao, J., Zhou, J., Li, X., Xia, M., Huang, Y., Huang, L., Wang, N., Zhang, X., Zhou, S., Wang, Y., et~al.: Usfm: A universal ultrasound foundation model generalized to tasks and organs towards label efficient image analysis. Medical Image Analysis  \textbf{96},  103202 (2024)

\bibitem{kendall2015bayesian}
Kendall, A., Badrinarayanan, V., Cipolla, R.: Bayesian segnet: Model uncertainty in deep convolutional encoder-decoder architectures for scene understanding. arXiv preprint arXiv:1511.02680  (2015)

\bibitem{NIPS2017_2650d608}
Kendall, A., Gal, Y.: What uncertainties do we need in bayesian deep learning for computer vision? In: Guyon, I., Luxburg, U.V., Bengio, S., Wallach, H., Fergus, R., Vishwanathan, S., Garnett, R. (eds.) Advances in Neural Information Processing Systems. vol.~30. Curran Associates, Inc. (2017)

\bibitem{lee2022speckle}
Lee, H., Lee, M.H., Youn, S., Lee, K., Lew, H.M., Hwang, J.Y.: Speckle reduction via deep content-aware image prior for precise breast tumor segmentation in an ultrasound image. IEEE Transactions on Ultrasonics, Ferroelectrics, and Frequency Control  \textbf{69}(9),  2638--2650 (2022)

\bibitem{software:nnj}
Miani, M., Warburg, F.: Nnj. GitHub. Note: https://github.com/IlMioFrizzantinoAmabile  (2023)

\bibitem{monteiro2020stochastic}
Monteiro, M., Le~Folgoc, L., Coelho~de Castro, D., Pawlowski, N., Marques, B., Kamnitsas, K., van~der Wilk, M., Glocker, B.: Stochastic segmentation networks: Modelling spatially correlated aleatoric uncertainty. Advances in neural information processing systems  \textbf{33},  12756--12767 (2020)

\bibitem{ning2021smu}
Ning, Z., Zhong, S., Feng, Q., Chen, W., Zhang, Y.: Smu-net: Saliency-guided morphology-aware u-net for breast lesion segmentation in ultrasound image. IEEE transactions on medical imaging  \textbf{41}(2),  476--490 (2021)

\bibitem{1661695}
Noble, J., Boukerroui, D.: Ultrasound image segmentation: a survey. IEEE Transactions on Medical Imaging  \textbf{25}(8),  987--1010 (2006). \doi{10.1109/TMI.2006.877092}

\bibitem{ovadia2019can}
Ovadia, Y., Fertig, E., Ren, C., Nado, Z., Sculley, D., Nowozin, S., Dillon, J., Lakshminarayanan, B., Snoek, J.: Can you trust your model's uncertainty? evaluating predictive uncertainty under dataset shift. Advances in Neural Information Processing Systems  (2019)

\bibitem{shapiro2003monte}
Shapiro, A.: Monte carlo sampling methods. Handbooks in operations research and management science  \textbf{10},  353--425 (2003)

\bibitem{thomas2006elements}
Thomas, M., Joy, A.T.: Elements of information theory. Wiley-Interscience (2006)

\bibitem{venturini2025whole}
Venturini, L., Budd, S., Farruggia, A., Wright, R., Matthew, J., Day, T.G., Kainz, B., Razavi, R., Hajnal, J.V.: Whole examination ai estimation of fetal biometrics from 20-week ultrasound scans. NPJ Digital Medicine  \textbf{8}(1),  1--12 (2025)

\bibitem{wright2023fast}
Wright, R., Gomez, A., Zimmer, V.A., Toussaint, N., Khanal, B., Matthew, J., Skelton, E., Kainz, B., Rueckert, D., Hajnal, J.V., et~al.: Fast fetal head compounding from multi-view 3d ultrasound. Medical Image Analysis  \textbf{89},  102793 (2023)

\bibitem{xu2021exploiting}
Xu, L., Gao, S., Shi, L., Wei, B., Liu, X., Zhang, J., He, Y.: Exploiting vector attention and context prior for ultrasound image segmentation. Neurocomputing  \textbf{454},  461--473 (2021)

\bibitem{10.1007/978-3-031-72111-3_33}
Zepf, K., Wanna, S., Miani, M., Moore, J., Frellsen, J., Hauberg, S., Warburg, F., Feragen, A.: Laplacian segmentation networks improve epistemic uncertainty quantification. In: Linguraru, M.G., Dou, Q., Feragen, A., Giannarou, S., Glocker, B., Lekadir, K., Schnabel, J.A. (eds.) Medical Image Computing and Computer Assisted Intervention -- MICCAI 2024. pp. 349--359. Springer Nature Switzerland, Cham (2024)

\bibitem{ZHANG2024102996}
Zhang, S., Metaxas, D.: On the challenges and perspectives of foundation models for medical image analysis. Medical Image Analysis  \textbf{91},  102996 (2024), \url{https://www.sciencedirect.com/science/article/pii/S1361841523002566}

\bibitem{zou2023review}
Zou, K., Chen, Z., Yuan, X., Shen, X., Wang, M., Fu, H.: A review of uncertainty estimation and its application in medical imaging. Meta-Radiology  \textbf{1}(1),  100003 (2023)

\end{thebibliography}
%


\end{document}